\documentclass[a4paper,UKenglish,cleveref, autoref, thm-restate]{lipics-v2021}
\hideLIPIcs
\usepackage{algorithm}
\usepackage{algorithmic}
\usepackage{amsfonts}
\usepackage{subcaption}
\usepackage{amsmath}
\usepackage{amsthm}
\usepackage{amssymb}
\usepackage{hyperref}
\usepackage{mathtools}
\usepackage{xcolor}
\usepackage{url}
\usepackage{setspace}
\usepackage{xspace}
\usepackage{soul}

\usepackage{cleveref}
\usepackage{thm-restate}

\usepackage[normalem]{ulem}
\newcommand{\fwidth}[2]{\textsf{fw}(#1, #2)}
\newcommand{\width}[1]{\textsf{width}(#1)}
\newcommand{\pw}[1]{\textsf{par-width}(#1)}
\newcommand{\val}[1]{\textsf{Val}(#1)}

\newcommand{\ch}[1]{\textsf{Ch}_{#1}}
\newcommand{\NP}{$\mathsf{NP}$}

\newcommand{\mfdsizen}[2]{\mathsf{mfd}_\mathbb{N}(#1,#2)}
\newcommand{\mfdsizez}[2]{\mathsf{mfd}_\mathbb{Z}(#1,#2)}
\newcommand{\mfdsizeu}[2]{\mathsf{mfd}_{\{-1, 1\}}(#1,#2)}
\newcommand{\mfdsizeo}[2]{\mathsf{mfd}_{\{1\}}(#1,#2)}
\newcommand{\mfdsize}[2]{\mathsf{mfd}_\mathbb{Y}(#1,#2)}
\newcommand{\mfd}{\textsc{MFD}\xspace}
\newcommand{\mfdn}{\textsc{MFD}$_\mathbb{N}$\xspace}
\newcommand{\mfdz}{\textsc{MFD}$_\mathbb{Z}$\xspace}
\newcommand{\gpc}[1]{G_{P({#1})}}

\newcommand{\N}[0]{\mathbb{N}}
\newcommand{\Z}[0]{\mathbb{Z}}
\newcommand{\Y}[0]{\mathbb{Y}}

\nolinenumbers

\title{Width Parameters for Minimum Flow Decomposition}

\author{Andreas Grigorjew}{Universit\'e Paris-Dauphine, PSL University, CNRS UMR7243, LAMSADE, Paris, France}{}{0000-0003-0989-2415}{This work was supported by ANR project ANR-21-CE48-0022 (S-EX-AP-PE-AL). During the research stage of this work, A.G. was affiliated with the University of Helsinki.}
\author{Wanchote Po Jiamjitrak}{University of Helsinki, Finland}{}{0009-0005-9912-7558}{}
\author{Brendan Mumey}{Montana State University, United States}{}{0000-0001-7151-2124}{}
\author{Alexandru I. Tomescu}{University of Helsinki, Finland}{}{0000-0002-5747-8350}{}

\authorrunning{Grigorjew et al.} 

\Copyright{Andreas Grigorjew, Wanchote Po Jiamjitrak, Brendan Mumey, Alexandru I. Tomescu} 



\ArticleNo{XX}

\funding{Co-funded by the European Union (ERC, SCALEBIO, 101169716). Views and opinions expressed are however those of the author(s) only and do not necessarily reflect those of the European Union or the European Research Council. Neither the European Union nor the granting authority can be held responsible for them. This work has also received funding from the Research Council of Finland grants No.~346968, 358744, as well from US National Science Foundation award No.~2309902.\flag{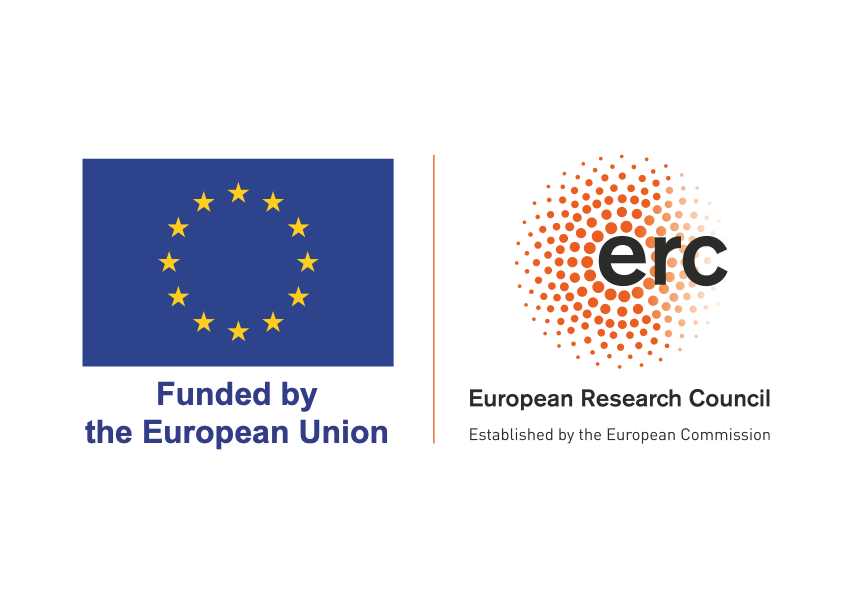}}

\begin{document}

\maketitle

\begin{abstract}
    Minimum flow decomposition (\mfd) is the strongly \NP-hard problem of finding a smallest set of integer weighted $s$-$t$ paths in an $s$-$t$ DAG $G$ whose weighted sum is equal to a given flow $f$ on $G$. Despite its many practical applications, we lack an understanding of graph structures that make \mfd easy or hard. 
	Recent progress is due to C{\'a}ceres et al.~[ACM TALG 2024], who showed that the DAG \emph{width}, the minimum number of paths to cover all edges, plays an essential role in the approximation of the problem.

    
    Our first set of results regard the computational complexity of \mfd parameterised by the width. This question was previously open, because \mfd on width-1 DAGs (i.e.~paths) is trivially solvable, and the existing $\mathsf{NP}$-hardness proofs use DAGs of unbounded width. We show that \mfd on width-2 DAGs is already \NP-hard and that \mfd on width-3 DAGs is \textit{strongly} \NP-hard.

	Our main contribution complements these hardness bounds, as we show that \textit{weak} \NP-hardness is the best we can hope for on width-2 DAGs. In fact, we prove the more general statement that \mfd with unary coded input can be solved in quasi-polynomial time on DAGs of constant parallel-width, which includes width-2 DAGs. The parallel-width of a DAG $G$ ($\pw{G}$) was defined by Deligkas and Meir~[MFCS 2017] as the size of the largest minimal $s$-$t$ cut-set.
    We obtain these results by, a) interpreting flow decompositions as a sequence of certain digraph minor operations defined by Deligkas and Meir [MFCS 2017], and b) defining a new notion of width of a flow network, \textit{flow-width} of $(G,f)$, defined as the minimum number of paths covering all edges of $G$, where every edge $e$ can be covered by at most $f(e)$ paths. Using (a) and (b), we show as an intermediate result, an improved upper bound $(\lfloor\log \Vert f\Vert\rfloor+1) \cdot \pw{G}$ for \mfd, where $\Vert f\Vert$ is the largest flow weight of all edges. 

\end{abstract}


\newpage 

\section{Introduction}
Minimum Flow Decomposition (\mfd) of a given flow on a directed graph is the problem of finding a minimum sized set of weighted paths, such that for every edge, the sum of the weights of the paths crossing this edge is equal to the flow on the edge. It is a standard result that every flow on a directed acyclic graph can be decomposed to at most $m-n+2$ weighted paths~\cite{ahyja1993network,vatinlen2008simple}, where $n$ is the number of nodes and $m$ is the number of edges. In this paper, all graphs will be directed acyclic graphs (DAGs) that are multigraphs (that is, we allow parallel edges) with single source $s$ and single sink $t$, and we use the term \emph{flow network} to refer to a pair $(G, f)$ of a DAG $G$ and a flow $f$ on $G$. This paper addresses the variant \mfdn, where the flow and path weights are non-negative integers. We distinguish it from \mfdz, where the flow and path weights can also be negative; we use \mfd to refer to both variants and we write $\mfdsize{G}{f}$ for the size of the minimum flow decomposition using path weights in a set $\Y$.

\mfd is strongly \NP-hard, even on DAGs, and even when the flow values come from $\{1, 2, 4\}$ \cite{hartman2012split}. However, the problem has a wide range of applications, e.g. in Bioinformatics~\cite{bernard2014efficient, dias2023safety, gatter2019ryuto, pertea2015stringtie, tomescu2013novel, tomescu2015explaining, williams2019rna, baaijens2020strain,shao2017theory}, transportation planning~\cite{olsen2022study} or network design~\cite{hartman2012split, mumey2015parity, vatinlen2008simple}. 
Despite this widespread use, we lack a good theoretical understanding behind the complexity of \mfdn.
Moreover, as opposed to other \NP-hard problems which have been extensively studied on restricted classes of graphs, we have only little knowledge of graph classes that make the problem tractable or any other structural properties that make \mfdn easier. Below, we describe in more detail the current state of the art around \mfd.

\subsection{Related work} Kloster et al.~\cite{kloster2018practical} showed that \mfdn is in linear FPT time $2^{O(k^2)}\cdot(n + \log\Vert f\Vert)$, where the parameter $k = \mfdsizen{G}{f}$ is the size of the optimal solution and $\Vert f\Vert$ is the maximum norm of $f$ (i.e., the largest flow weight of all edges), however, it is not known whether this time complexity is tight. \mfdn also admits polynomially-sized Integer Linear Programming formulations~\cite{dias2022fast,grigorjew2023accelerating}. It is also known that \mfdn is $\mathsf{APX}$-hard (i.e., for some $\varepsilon > 0$ there is no $(1+\varepsilon)$-approximation unless $\mathsf{P}=\mathsf{NP}$)~\cite{hartman2012split}, and there exists an approximation algorithm that decomposes all but an $\varepsilon$ fraction of the flow with a factor of $O(1/\varepsilon^2)$~\cite{hartman2012split}. 

Recent theoretical progress is due to Cáceres et al.~\cite{caceres2024width}, who showed that the \textit{width} of the graph, namely, the minimum number of $s$-$t$ paths to cover all edges, play an important role in the decomposition of flows. The width is a natural lower bound to \mfd, since every flow decomposition is also such a path cover. For example, using width, they improved the approximation factor lower bound of the widely-used greedy approach for \mfdn\xspace-- to iteratively remove the currently highest weighted $s$-$t$ path in the graph~\cite{vatinlen2008simple,gatter2019ryuto,tomescu2013novel,bernard2014efficient,shao2017theory} -- to $\Omega(m/\log m)$ in the worst case~\cite{caceres2024width}. To obtain this, they exploited the fact that the width can grow exponentially after subtracting flow and removing edges of zero flow.

To partially address this issue, Cáceres et al.~\cite{caceres2024width} proved that restricting the input to \textit{width-stable} graphs, which Cáceres et al.~\cite{caceres2024width} defined as the class of $s$-$t$ DAGs $G$ whose $s$-$t$ subgraphs all have smaller or equal width to $G$, improves the approximation factor of the greedy algorithm to $O(\log\val{f})$, where $\val{f}$ is defined as the sum of the flow that leaves $s$.
A notable example of width-stable DAGs are \textit{series-parallel} DAGs; a numerous amount of $\mathsf{NP}$-hard problems are easier to solve on them, see e.g.~\cite{eppstein1992parallel, valdes1979recognition}. Despite such better approximation factor for greedy on width-stable graphs, \mfd remains strongly $\mathsf{NP}$-hard even in their subclass of series-parallel graphs~\cite{vatinlen2008simple}, and it is not known whether it can be decided in polynomial time whether a graph is width-stable.
Moreover, Cáceres et al.~\cite{caceres2024width} used the width to define an approximation algorithm for \mfdz of ratio $O(\log \Vert f\Vert)$ by repeatedly decomposing the flow using $\width{G} \leq \mfdsizez{G}{f}$ paths per iteration.

As we show in this work, \mfd is related to the \textsc{generating set} problem, in which we are given a set $A$ of positive integers and a parameter $k$, and the goal is to find set of positive integers $Z$ of size $k$, such that every integer in $A$ is the sum of a subset of $Z$. It can be solved in time $O^*(s^{\log s})$\footnote{The notation $O^*(\cdot)$ hides factors that are polynomial in the input size.}, as Collins et al.~\cite{COLLINS2007129} observe, the powers of two of a most $s$ are a satisfiable set for $Z$. The number of sets of size $\log s$ to verify is upper bounded by $s^{\lceil\log s\rceil}$, and each set can be verified by solving subset sum on each input integer. No faster exact algorithms are currently known. However, heuristic algorithms exist~\cite{lozano2016genetic, COLLINS2007129}.

\subsection{Our contributions}
\label{sec:contributions}
As opposed to other \NP-hard problems, we have little understanding of what makes \mfd easy to solve or hard. As such, in this paper we establish further connections between structural properties or parameters of the graph, and the tractability status of \mfdn in relation to these parameters.
We first show that instances of small width remain \NP-hard. This was previously open, because all reductions showing the hardness of \mfd use instances of arbitrary widths.
We then show how using two related parameters, the \textit{flow-width} of a DAG $G$ and a flow $f$, which we introduce in \Cref{sec:parity-fixing-algorithm}, and the \textit{parallel-width} $\pw{G}$, recently introduced by Deligkas and Meir~\cite{deligkas2017directed}, can make \mfdn easier to solve.

Investigating width parameters is motivated by the following. On width-1 graphs, \mfd is trivially solvable. Moreover, we observe that some real-world instances from Bioinformatics~\cite{tung2020top} admit small widths. Further, we reveal structural features of flow decompositions, as we demonstrate in the following.

\subparagraph{Improved hardness results.}
In \Cref{sec:mfd-hardness} we show that \mfdn remains \NP-hard on small width instances. It is known that it remains $\mathsf{NP}$-hard on width-stable graphs, since the $\mathsf{NP}$-hardness reduction from~\cite{vatinlen2008simple} uses series-parallel graphs, which are a subclass of width-stable graphs~\cite{caceres2024width}. However, their width is not bounded and it remained open whether having a bounded width makes the problem easier. We address this by proving the following.

\begin{restatable}{theorem}{thmconstantwidthhard}
    \mfdn on a DAG $G$ is strongly $\mathsf{NP}$-hard even when $\width{G} = 3$.
    \label{thm:constant-width-hard}
\end{restatable}

To prove this theorem, we devise a construction where, even if the initial graph width is $3$, the width can grow arbitrarily large after removing a path of any MFD.
Our next theorem shows that the problem remains $\mathsf{NP}$-hard on width-2 graphs when the flow is coded in binary with a reduction from the \textsc{generating set} problem~\cite{COLLINS2007129}.

\begin{restatable}{theorem}{thmconstantparwcomplexity}
    \mfdn on a DAG $G$ is $\mathsf{NP}$-hard when $f$ is coded in binary, even when $\width{G} = 2$.
    \label{thm:constant-parw-binary}
\end{restatable}

We will also show that this result is tight by addressing the case when the flow is coded in unary (i.e., the natural number $a$ is stored using $a+1$ bits) in the following.

\subparagraph{Improved parameterised runtime.}
Our main result regards an improved parameterised runtime upper bound for \mfd. We use the parallel-width of $G$ (written $\pw{G}$),
which is defined as the size of the largest minimal cut-set of $G$ and was introduced in the context of a directed graph minor notion on $s$-$t$ digraphs~\cite{deligkas2017directed}. The main result of the paper is the following.
\begin{restatable}{theorem}{corconstantparwunary}
    Let $c\in\N$ be a constant. \mfdn on $(G,f)$ can be solved in time $\Vert f\Vert^{O(\log\Vert f \Vert)} \cdot n$ when $\pw{G}\leq c$, which is quasi-polynomial if $f$ is coded in unary. \label{cor:constant-parw-unary}
\end{restatable}
Note that, contrary to graphs of width at least 3, all width-2 graphs admit a parallel-width of 2.
\Cref{thm:constant-width-hard}, \Cref{thm:constant-parw-binary} and \Cref{cor:constant-parw-unary} provide a tight description of \mfdn's computational complexity on constant width. Moreover, the results show that \mfdn on constant parallel-width naturally generalizes the \textsc{generating set} problem. To the best of our knowledge, this is the first time that the parallel-width has been used in the analysis of the computational complexity of a problem.

Kloster et al.~\cite{kloster2018practical} stated the open problem, whether there exists a $O^*(2^{o(k^2)})$ time FPT algorithm for \mfd. Our results make partial progress on this question, as the reduction used for \Cref{thm:constant-parw-binary} shows that the existence of such an algorithm implies that \textsc{generating set} can be solved in time $O^*(2^{o((\log s)^2)}))$. This is an unsolved problem, as previous solvers for \textsc{generating set} do not resolve this question~\cite{lozano2016genetic, COLLINS2007129}.

\subparagraph{Flow-width and graph minors.}
We now explain how we achieve \Cref{cor:constant-parw-unary}.
Our main tool is to analyse the flow network structure in \Cref{sec:parity-fixing-algorithm}, using the \emph{flow-width} $\fwidth{G}{f}$, a new parameter that we define and which might be of independent interest. In contrast to \mfdz, the DAG width faces the drawback of covering an edge too often: given a positive flow $f$ on $G$, it might be necessary for all minimum path covers to cover an edge more than $f(e)$ times. The flow-width uses the given flow $f$ as upper bounds for the number of times covering paths can use edges, and satisfies $\width{G} \leq \fwidth{G}{f}$ for all positive flows $f$. Moreover, the flow-width is still a natural lower bound to the optimal solution of \mfdn.

We begin by showing that the flow-width is monotone on the flow, if and only if the underlying DAG is width-stable, and that the flow-width is a more fine-grained parameter to analyse the width-stability compared to width alone. This implies the following corollary, improving the best known approximation ratio of \mfdn on width-stable DAGs from $O(\log \val{f})$~\cite{caceres2024width} to $O(\log \Vert f\Vert)$:

\begin{restatable}{corollary}{thmapproxalgows}
    On width-stable DAGs we can approximate \mfdn given the input $(G, f)$ with ratio $\lfloor\log \Vert f\Vert\rfloor+1$ in time $O(nm\log\Vert f\Vert)$.
    \label{thm:approx-algo-ws}
\end{restatable}

Next, we argue that the graph minors defined by Deligkas and Meir~\cite{deligkas2017directed} correspond to possible edge saturations of flow decompositions in \Cref{sec:approx-alg}. This enables us to use the parallel-width as parameter, which was defined in the context of these graph minors.
The flow-width plays a central role in this context, as it connects the width and the parallel-width:

\begin{restatable}{lemma}{lempwidthwidthfwidth}
For all $s$-$t$ DAGs $G$, the following equalities hold:
\begin{align*}
    \width{G} & = \min\{\fwidth{G}{f} \mid f > 0 \},\\
    \pw{G} & = \max\{\fwidth{G}{f} \mid f\geq 0 \}.
\end{align*}
\label{lem:pwidth-width-fwidth}
\end{restatable}
While $\pw{G}$ is \NP-hard to compute, it admits a polynomial-time algorithm if it is constant~\cite{deligkas2017directed}.
Using \Cref{lem:pwidth-width-fwidth}, we obtain a new upper bound for \mfdn, which we will use to prove \Cref{cor:constant-parw-unary} and is of independent interest:

\begin{restatable}{lemma}{thmapproxalgogeneral}
    Given a flow network $(G,f)$ with $f>0$, we can decompose the flow using at most $\pw{G} \cdot (\lfloor\log\Vert f\Vert\rfloor+1)$ paths in time $O(nm\log\Vert f\Vert)$.
    \label{thm:approx-algo-general}
\end{restatable}

\section{Preliminaries}

By default, graphs $G = (V(G), E(G))$ are assumed to be acyclic digraphs with a single source $s$ and a single sink $t$ (an \emph{$s$-$t$ DAG}), and all subgraphs of $G$ that we consider are $s$-$t$ DAGs. We use $n$ and $m$ to denote the number of vertices and edges, respectively, and we denote by $\deg^+(v)$ and $\deg^-(v)$ the out- and in-degree of a vertex $v$, respectively. A set of edges $A$ is called a cut-set if every $s$-$t$ path includes an edge of $A$. We denote by $[k]$ the set $\{1,\dots,k\}$. We call functions $f:E(G)\to\Y$ \textit{pseudo-flows}\footnote{Commonly in the literature, (pseudo-)flows are also required to be skew-symmetric and to be upper bounded by some capacity function on the edges. These properties play no role in this article.} on $G$, where $\Y\in\{\N,\Z\}$. We use the notation $f + g$ and $\mu f$ to denote the element wise sum of pseudo-flows and scalar multiplication. The value $0$ may (depending on context) denote the pseudo-flow that is equal $0$ on every edge. We write $f\leq g$ (resp.\ $f<g$) to denote $f(e)\leq g(e)$ (resp.\ $f(e) < g(e)$) for every edge $e\in E(G)$ and two pseudo-flows $f,g$ on $G$.

A \textit{flow} on $G$ is a pseudo-flow whose internal vertices $V\setminus\{s,t\}$ satisfy the flow conservation property (incoming flow is equal to the outgoing flow). The sum of two flows $f + g$, the multiplication $\mu f$ of a flow $f$ with a scalar $\mu$ and the pseudo-flow $0$ are flows. We denote by the flow value $\val{f}$ of $f$ the sum of the outgoing flow of $s$, and we call $f(e)$ the weight of edge $e$. We call the pair $(G, f)$ of an $s$-$t$ DAG $G$ and a flow $f$ a flow network. Given an $s$-$t$ path $P$, we also denote by $P$ the indicator flow of the path, that is, $P(e) = 1$ for all $e\in P$ and $P(e) = 0$ otherwise. We say that a path $P$ in $(G, f)$ carries flow $\mu$ if $f(e) \geq \mu$ for all $e \in P$, and we define that any $v-v$ path carries flow $\infty$. We also say that $f$ covers an edge $e$ (or a set of edges $A$) if $f(e) \neq 0$ (for all $e \in A$). Given lower- and upper bounds $L,R:E \to \Y\cup\{\infty\}$ , we call a flow $f:E \to \Y$ a minimum flow with respect to $L$ and $R$, if $L \leq f \leq R$ and the value of $f$ is minimised. If $L(e) \in \{0,1\}$ and $R(e)=\infty$ for all $e\in E$, we call $f$ a minimum covering flow.

\begin{definition}
    Given a flow $f:E(G)\to\Y$ on $G$, a \textit{flow decomposition} of size $k$ of $f$ is a family of $s$-$t$ paths $\mathcal{P} = (P_1,\dots,P_k)$ and weights $(w_1,\dots,w_k)\in\Y^k$ such that $f = w_1P_1 + \dots + w_kP_k$.
\end{definition}

\begin{definition}
    For a flow $f$ on $G$, let $\mfdsize{G}{f}$ be the size of a smallest flow decomposition of $f$ using weights in $\Y$. We denote by \mfdn and \mfdz the problems of finding a flow decomposition of smallest size for $\Y = \N$ and $\Y = \Z$, respectively.
\end{definition}

We let $\Vert f\Vert = \max_{e \in E} |f(e)|$ denote the maximum norm on flows, and write $f\equiv_2 g$ if and only if the flows $f$ and $g$ have the same parity on all edges, that is, for all edges $e\in E(G)$, $f(e)$ is odd if and only if $g(e)$ is odd. An important tool we use to analyse graph structure is the \textit{width}:
\begin{definition}
    \label{def:width}
    We define $\width{G}$ as the minimum number of $s$-$t$ paths in a DAG $G$ needed to cover all edges in $E(G)$.
\end{definition}
\Cref{def:width} is a variant of the more common problem of finding a minimum number of paths to cover all vertices.
Sets of paths and flows are equivalent in the sense that one can be transformed into the other. Given a set of $s$-$t$ paths $\mathcal{P}$ on $G$, we can define a unique flow $f_\mathcal{P} = \sum_{p\in\mathcal{P}} p$ that counts the number of paths on every edge. Conversely, given a flow $f:E(G)\to\N$, we can define a path cover $\mathcal{P}_f$ by the \textit{Flow Decomposition Theorem}~\cite{ahyja1993network}, simply by decomposing the flow into weight $1$ paths. It holds that $\width{G} = \min \{ \val{f} \mid f(e) > 0\, \forall e\in E(G)\}$ is equal to the value of a minimum covering flow. We say that a path cover induces a flow and vice versa.
We can find such a minimum covering flow in $O(m^{1+o(1)})$ time~\cite{chen2025maximum}, and find the induced path cover in additional time $O(\width{G} \cdot n)$, since in DAGs, every path is of length at most $n-1$.

In this paper, we show a connection of \mfdn to the following problem.
\begin{definition}
    In the \textsc{generating set} problem, we are given a set of positive integers $A=\{a_1,\dots,a_n\}$ and a positive integer $k$, and we want to decide if there is a set of positive integers $Z=\{z_1,...,z_k\}$ such that every element $a \in A$ is the sum of some subset of the elements in $Z$.
    \label{def:gen-set}
\end{definition}
The \textsc{generating set} problem is known to be $\mathsf{NP}$-hard \cite{COLLINS2007129} and the proof of that uses integers that grow exponentially. 

\section{\mfd hardness results}
\label{sec:mfd-hardness}

In this section, we analyse the parameterised complexity of \mfdn using the width. DAGs $G$ of $\width{G} = 1$ are $s$-$t$ paths, which obtain a unique flow decomposition. However, all known hardness reductions use arbitrarily large widths.
A previous reduction \cite{vatinlen2008simple} showed that width-stable \mfdn instances are strongly $\mathsf{NP}$-hard. 

\thmconstantwidthhard*
\begin{proof}
    \begin{figure}[h]
        \centering
        \includegraphics[scale= 0.85, trim= 40 270 300 80,clip]{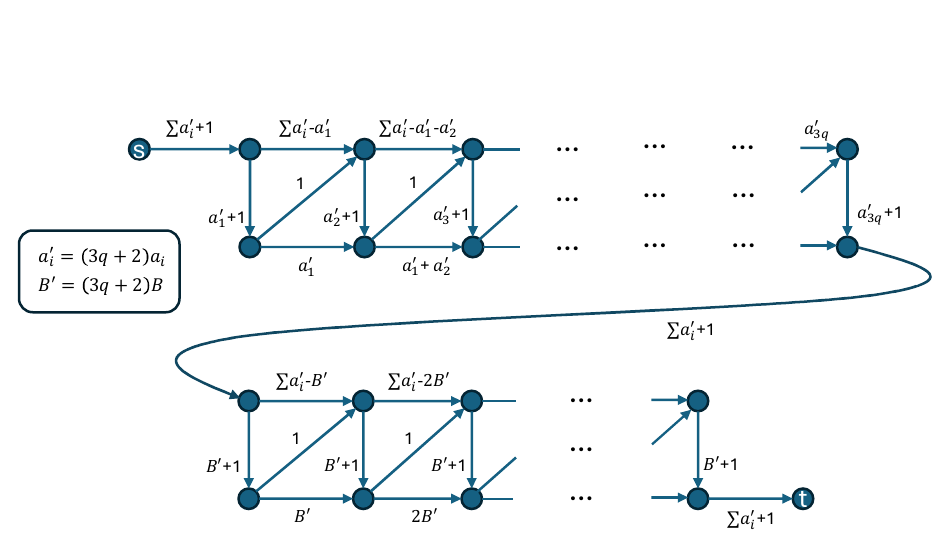}
        \caption{An \mfdn instance of width 3 for solving \textsc{3-partition} problem.}
        \label{fig:width-3-mfd-3-partition}
    \end{figure}
    We will show a reduction from the \textsc{3-partition} problem to \mfdn on $G$ of width 3. Let $a_1,\dots,a_{3q},B \in \N$ such that $a_i \in (B/4,B/2)$ (with this assumption, \textsc{3-partition} remains \NP-hard~\cite{michael1979computers}). We want to find a partition of the $a_1, \dotsc, a_{3q}$ to sets $S_1, \dotsc, S_q$ such that, for all $i$, $\sum_{a \in S_i} a = B$. Note that this restriction implies that $\forall i: |S_i|=3$.
    
    Consider the \mfdn instance $G$ in \Cref{fig:width-3-mfd-3-partition}, which we can divide into the top and the bottom components as illustrated. Each vertical edge in the top component represents an $a_i$ and each vertical edge in the bottom component represents $B$. We claim that there is a solution to the \textsc{3-partition} problem if and only if this \mfdn instance on $G$ has a solution of size $3q+1$.

    First, we show that, if we have a solution to the \textsc{3-partition} problem, then we can decompose $G$ into $3q+1$ weighted paths. This can be done by routing one path of weight 1 to saturate all the diagonal edges. Then, for each $a_i \in S_j$, we route a path of weight $(3q+2)a_i$ through the $i^{th}$ vertical edge in the top component and through the $j^{th}$ vertical edge in the bottom component.

    Now, we will show that, if we can decompose $G$ into $3q+1$ weighted paths, then we can construct a solution to the \textsc{3-partition}. Let $F = \{(P_i, w_i) \mid i \in [3q+1]\}$ be the set of weighted paths in the solution in a non-decreasing order of $w_i$. Let $U \subseteq F$ be the set of paths with a weight of 1 each. Note that all the diagonal edges must be saturated by $U$. Since $|U| \leq 3q+1$ and each vertical edge in the top component has a flow of at least $3q+2$, they are not saturated by $U$. This means that we need at least $3q$ paths to saturate all vertical edges in the top component, or in other words, $|F\setminus U| \geq 3q$. Hence, $|U|=1$, and a path of weight $1$ routes through all diagonal edges in both the top and bottom components. For the remaining $3q$ paths, since each path can only pass one vertical edge in the top components, it must saturate one vertical edge in the top component. To construct a solution to the \textsc{3-partition} problem, we put $a_i$ in a set $S_j$ if the path that saturates the $i^{th}$ vertical edge in the top component uses the $j^{th}$ vertical edge in the bottom component.
\end{proof}

The instances that we use to prove \Cref{thm:constant-width-hard} are not width-stable because after removing the diagonal edges using a single path of weight $1$, we obtain the \mfdn instances of the known \textsc{3-partition} reduction of width $3q$~\cite{vatinlen2008simple}.
We now show a reduction from \textsc{generating set} to show weak \NP-hardness of \mfdn.

\thmconstantparwcomplexity*
\begin{proof}
\begin{figure}[h]
    \centering
    \includegraphics[scale= 0.9, trim= 40 290 300 240,clip]{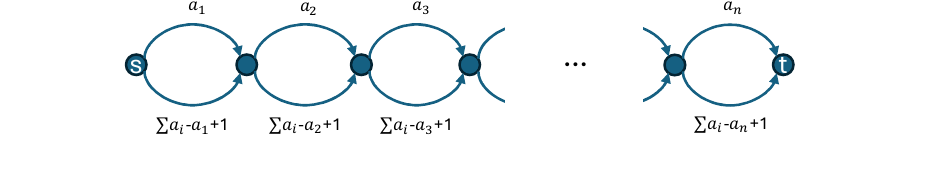}
    \caption{An \mfdn instance of width 2 for solving \textsc{generating set} problem. Note that this multigraph can be turned into a graph by subdividing each edge.}
    \label{fig:width-2-mfd-genset}
\end{figure}
We will show a reduction from the \textsc{generating set} problem to \mfdn on $G$ of width 2.
Consider a \textsc{generating set} instance $\{ a_1, \dots, a_n \}$ and the \mfdn instance in \Cref{fig:width-2-mfd-genset}. Let $e^{(t)}_i$ and $e^{(b)}_i$ be the $i^{th}$ top and bottom edge from the left, respectively. We construct $G$ of width $2$ by using $a_i$ as a weight of $e^{(t)}_i$. We also let the total ($s$-$t$)-flow have a value of $\sum_{i=1}^n a_i+1$. We will show that \textsc{generating set} problem has a solution of size $k$ if and only if \mfdn on $G$ has a solution of size $k+1$.

First, we show that, if we have a solution of \textsc{generating set} of size $k$, we can obtain a solution of \mfdn of size $k+1$. Let $Z=\{z_1,...,z_k\}$ be a solution of \textsc{generating set}. We have that, for all $i \in [n]$, there is $\chi_{ij}\in \{0,1\}^{n \times k}$ such that $a_i = \sum_{j=1}^k \chi_{ij} z_j$. Let w.l.o.g. $\sum_{i=1}^n \chi_{ij} \geq 1$ for all $j \in [k]$. In the corresponding \mfdn solution, for $j \in [k]$, we route a path $P_j$ of weight $z_j$ via $e^{(t)}_i$ when $\chi_{ij}=1$, and route $P_j$ via $e^{(b)}_i$ when $\chi_{ij}=0$. After we route $P_1, \dots, P_k$, all the top edges will be saturated. Finally, we route $P_{k+1}$ of remaining weight via all bottom edges.

Next, we show that, if we have a solution of \mfdn of size $k+1$, we can obtain a solution of \textsc{generating set} of size $k$. Let $\{(P_1,w_1),...,(P_{k+1},w_{k+1})\}$ be the \mfdn solution. Note that the total flow in this instance has weight $\sum_{i \in [n]}a_i+1$, and the total of the weight of all $P_j$ that use at least one top edge is at most $\sum_{i \in [n]} a_i$. Since the total flow is strictly more than the total weight of all $P_i$ that use at least one top edge, there must be one path in our solution that uses only bottom edges. W.l.o.g, Let $P_{k+1}$ be such a path. Notice that, for the remaining $k$ paths $P_1,..., P_k$, if all their weights are distinct, we claim that $\{w_1,...,w_k\}$ is a solution of \textsc{generating set}. This can be done by setting $\chi_{ij}=1$ when $P_j$ is routed via $e^{(t)}_i$, and $\chi_{ij}=0$ when $P_j$ is routed via $e^{(b)}_i$. Now, when their weights are not distinct, we have a multiset that satisfies the \textsc{generating set} problem. We can turn this multiset into a set by the following. Let $w$ be the smallest positive integer in this multiset and $d$ be the number of copies. We replace these integers with $w, 2w, ..., dw$ and adjust $\chi$ accordingly. We can repeat this process until all integers are distinct. The process is terminated in at most $k$ rounds since we obtain at least one distinct element $w$ after each round. 
\end{proof}

%
%

In the following sections, we show that \mfdn can, like \textsc{generating set}, be solved in quasi-polynomial time, \mfdn if the instance is of width 2 and has a unary coded integer input. More precisely, we will show a more general result using the \textit{parallel-width} (see \Cref{def:parallel-width}) introduced by Deligkas and Meir~\cite{deligkas2017directed}.

\section{Decomposing via flow-width}
\label{sec:parity-fixing-algorithm}
Suppose we are decomposing a flow path by path. During this decomposition, flow is subtracted from edges until they are saturated and cannot be traversed by further weighted paths any more. We mitigate oversaturating an edge by introducing the flow-width of a flow network. We then reformulate a heuristic given by Mumey et al.~\cite{mumey2015parity} using the flow-width.

\subsection{Flow-width}
\label{sec:flow--width}
While $\width{G}$ is always a lower bound to $\mfdsizen{G}{f}$ for flows $f>0$, we have the problem that minimum path covers might have to cover an edge $e$ at least $\mu > f(e)$ times (see \Cref{fig:minimal-minimum} for an example). A more accurate lower bound to \mfdn is thus a minimum path cover that respects upper bounds defined by flows.

\begin{figure}
    \centering
    \includegraphics[width=0.4\textwidth]{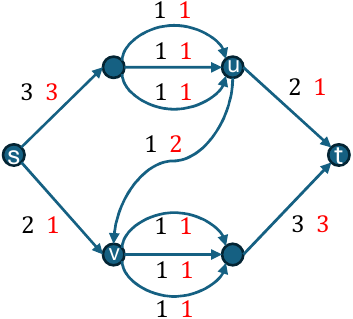}
    \caption{An example of minimally covering flow (black) and minimum covering flow (red). Note that the value of the black flow $(5)$ is larger than the value of the red flow $(4)$, despite being smaller in the central edge.}
    \label{fig:minimal-minimum}
\end{figure}

\begin{definition}
    Let $G$ be a DAG and $f$ be an integral non-negative flow on $G$. We define the \textit{flow-width} of $G$ and $f$, $\fwidth{G}{f}$, as the smallest number of paths satisfying the following properties:
    \begin{description}
        \item[1. Covering] Every edge $e\in E(G)$ with $f(e) > 0$ appears in at least one path, and
        \item[2. Upper bounds] Every edge $e\in E(G)$ appears in at most $f(e)$ paths.
    \end{description}
    \label{def:flow-width}
\end{definition}

Similarly to computing the width, the flow-width $\fwidth{G}{f}$ can be computed by finding decrementing paths from $f$, or by finding a minimum-cost flow. Closely related to the flow-width are \textit{minimally covering flows}:
\begin{definition}
    A non-negative flow $f$ is minimally covering if all $s$-$t$ paths in $G$ carry at most flow $1$.
\end{definition}
Equivalently, a non-negative flow $f$ is minimally covering if every flow $g\leq f$ that covers the same set of edges is equal to $f$. In other words, $\fwidth{G}{f}$ is the value of a minimally covering flow.

The flow-width can be applied to flow networks whose flow has weight $0$ on some edges. These edges are excluded from the covering, and we only work with the edges that have positive flow. 
For the analysis of relevant graph structure, we consider the following class of subgraphs.
\begin{definition}[\cite{caceres2024width}]
    Let $G$ be an $s$-$t$ DAG and $f:E\to\mathbb{N}$ a flow. The \textit{flow-subgraph} $G|_f = (V|_f, E|_f)$ of $G$ is defined by $V|_f = V \setminus \{ v \in V \mid \sum_{u:(u,v)\in E(G)} f(u,v) = \sum_{u:(v,u)\in E(G)} f(v,u) = 0 \}$ and $E|_f = \{ e \in E \mid f(e) > 0 \}$.
\end{definition}

\begin{lemma}[Properties of flow-widths]
    Let $G$ be an $s$-$t$ DAG.
    \begin{enumerate}
        \item For all flows $f \geq 0$, there exists a minimally covering flow $g$ with $0 \leq g \leq f$, $G|_f = G|_g$ and $\val{g} = \fwidth{G}{f}$, \label{lem:mcf-always-exists}
        \item For all flows $f \geq 0$, $\width{G|_f}\leq\fwidth{G}{f}\leq\mfdsizen{G}{f}$, \label{lem:fw-in-between}
        \item For all flows $g\geq f \geq 0$ with $G|_g = G|_f$, $\fwidth{G}{f} \geq \fwidth{G}{g}$. \label{lem:fw-grows-for-smaller-flows}
    \end{enumerate}
    \label{lem:fw-first-properties}
\end{lemma}
\begin{proof}
    We prove all given statements one by one.
    \begin{enumerate}
        \item We show that such a flow $g$ can be found by solving the following linear program (LP):
        \begin{equation}
        \begin{aligned}
        &\min \val{g}, \text{subject to}\\
        &g \text{ is a flow on $G$},\\
        &0 \leq g \leq f,\\
        &g(e) \geq 1 \text{ for all $e\in E(G)$ where $f(e) \geq 1$}.\\
        \label{eq:minally-covering-flow-Gf}
        \end{aligned}
        \end{equation}
        Note that the LP is always feasible, as $f$ itself satisfies all the conditions. Moreover, every optimal solution $g$ to the LP is an integer flow as $f$ is an integer flow. Finally, every optimal solution $g$ is a minimally covering flow, because if it is not, there exists an $s$-$t$ path carrying strictly more flow than $1$. Reducing the flow along this path until it only carries flow $1$, we obtain a flow of smaller value covering the same set of edges.
        \item The first inequality of Property \ref{lem:fw-in-between} holds because the upper bounds in the definition of the width are $\infty$ and thus larger than for the flow-width. The second inequality holds, because every positive flow decomposition is a path cover whose number of paths on every edge is upper bounded by the flow on that edge.
        \item To show that Property \ref{lem:fw-grows-for-smaller-flows} holds, note that because $G|_g = G|_f$, the same set of edges have to be covered for both flows, but the upper bounds defined by $f$ are stricter. This means that the set of path covers that satisfy the Covering and Upper bounds properties in \Cref{def:flow-width} for $f$ also satisfy them for $g$.
    \end{enumerate}
\end{proof}
In general, it can happen that $\width{G|_f} > \width{G|_g}$ for flows $f < g$. We thus consider the class of \textit{width-stable} DAGs, whose width does not grow when removing weighted paths from the flow:
\begin{definition}[\cite{caceres2024width}]
    The class of \textit{width-stable} DAGs is defined as all $G$ that satisfy $\width{G|_f} \leq \width{G|_g}$ for all flows $0\leq f\leq g$. 
\end{definition}

Width-stable DAGs have been characterised using funnels~\cite{garlet2020efficient}, which are DAGs that generalise in/out-forests:
\begin{definition}\
    \begin{enumerate}
        \item An $s$-$t$ DAG $F$ is called a \textit{funnel} if along any $s$-$t$ path the vertices $v$ first satisfy $\deg^-(v)\leq 1\leq \deg^+(v)$ and then $\deg^-(v)\geq 1\geq \deg^+(v)$. Equivalently, an $s$-$t$ DAG is a funnel, if every $s$-$t$ path has a \textit{private} edge, i.e., an edge, which is not contained in any other $s$-$t$ path~{\normalfont\cite{garlet2020efficient}}.
        \item For a flow-subgraph $F$ of an $s$-$t$ DAG $G$, we call a $u$-$v$ path $P_F$ in $G$ a \textit{central path}, if $F$ is a funnel, $u,v\in F$, and if inside $F$, $\deg^-(u) > 1$ and $\deg^+(v) > 1$. We write that $G$ contains a funnel subgraph $F$ with central path $P$~{\normalfont\cite{caceres2024width}}.
    \end{enumerate}
\end{definition}

In \Cref{fig:minimal-minimum}, the graph is a funnel with a $u$-$v$ central path consisting of a single edge. Note that this DAG without the edge $(u,v)$ is the funnel-subgraph $F$, while $(u,v)$ as the central path is \textit{not} part of $F$.

\begin{lemma}[\cite{caceres2024width}, Lemma 13]
    Let $G$ be an $s$-$t$ DAG. The following are equivalent.
    \begin{enumerate}
        \item $G$ is width-stable,
        \item For any flow $f \geq 0$ on $G$, there exists an $s$-$t$ path in $G|_f$ carrying flow $\val{f}/\width{G|_f}$,
        \item $G$ has no funnel subgraph with central path.
    \end{enumerate}
    \label{lem:wstable-funnel-char}
\end{lemma}

\Cref{lem:wstable-funnel-char} shows that the DAG in \Cref{fig:minimal-minimum} is not width-stable.
Conveniently, the width-stable property seamlessly extends to flow-widths. Indeed, the following property shows the impact that the graph structure has on the possible minimally covering flows that can be defined.

\begin{lemma}
    Let $G$ be an $s$-$t$ DAG.
    \begin{enumerate}
        \item If $G$ is width-stable, then $\fwidth{G}{f} = \width{G|_f}$ for all flows $f\geq 0$ on $G$. \label{lem:fwstable-1}
        \item $G$ is width-stable if and only if $\fwidth{G}{f} \leq \fwidth{G}{h}$ for all flows $0\leq f\leq h$. \label{lem:fwstable-2}
    \end{enumerate} \label{lem:fwstable}
\end{lemma}
\begin{proof}
    \begin{enumerate}
        \item Let $G$ be width-stable. Assume $f$ is a minimally covering flow on $G$, that is, $\val{f} = \fwidth{G}{f}$. By \Cref{lem:wstable-funnel-char} there exists an $s$-$t$ path in $G|_f$ carrying $\val{f}/\width{G|_f}$ flow. Because $f$ is minimally covering, the largest flow any $s$-$t$ path can carry is $1$.
        And thus, $\val{f}/\width{G|_f} \leq 1$ and $\fwidth{G}{f} = \val{f} \leq \width{G|_f}$. We also have $\fwidth{G}{f} \geq \width{G|_f}$ by \Cref{lem:fw-first-properties}.\ref{lem:fw-in-between}. It follows that $\fwidth{G}{f} = \width{G|_f}$.
        Let now $f\geq0$ be any (not necessarily minimally covering) non-negative flow on $G$. Let $h$ be a minimally covering flow on $G$ with $h \leq f$ and $G|_h = G|_f$, which exists by \Cref{lem:fw-first-properties}.\ref{lem:mcf-always-exists}. Then, by \Cref{lem:fw-first-properties}.\ref{lem:fw-grows-for-smaller-flows}, $\fwidth{G}{f} \leq \fwidth{G}{h} = \width{G|_f}$. Since also $\fwidth{G}{f} \geq \width{G|_f}$, we have $\fwidth{G}{f} = \width{G}$.
        
        \item If $G$ is width-stable, then the statement follows from Statement \ref{lem:fwstable-1}. Assume $G$ is not width-stable, and let $F$ be a funnel with central path $P$. We define $h$ and $f$ in the following way.  Both flows are minimally covering, with flow $1$ on each of the private edges of the funnel. The flow $h$ routes one unit of flow along the central path, while $f$ does not, i.e., $\fwidth{G}{h} = \val{h} = \width{F}-1$ and $\fwidth{G}{f} = \val{f} = \width{F}$. We have $f \leq h+f$ and $\fwidth{G}{f} > \fwidth{G}{h} \geq \fwidth{G}{h+f}$ by \Cref{lem:fw-first-properties}.\ref{lem:fw-grows-for-smaller-flows}.
    \end{enumerate}
\end{proof}

\subsection{Parity fixing with minimally covering flows}
\label{sec:parity-fixing-minimal-flows}

Cáceres et al.~\cite{caceres2024width} used the width to define an approximation algorithm for \mfdz of ratio $O(\log \Vert f\Vert)$, where $\Vert f\Vert = \max_{e \in E} |f(e)| \leq \val{f}$.
This algorithm follows a ``parity fixing'' approach: it constructs a unitary flow
(that is, a flow with values in $\{-1, 0, 1\}$), 
which, when subtracted from $f$, yields a flow that is even everywhere, and they showed that all unitary flows can be decomposed into at most $\width{G}$ paths
with weights in $\{-1, 1\}$.
The resulting flow can then be divided by two, to repeat the procedure until the flow is zero. 
To sum up, Cáceres et al.~\cite{caceres2024width} proved that it is possible to express $f$ as
\begin{equation}
    \label{eq:parity-fixing}
    f = \sum_{i=0}^{\lfloor\log\Vert f\Vert\rfloor} 2^i f_i, \text{~~~with~~~}\mfdsizeu{G}{f_i} \leq \width{G}.
\end{equation}

We now present the parity-fixing heuristic for \mfdn by Mumey et al.~\cite{mumey2015parity} in order to theoretically analyse its performance.
Given a flow $f \geq 0$, it iteratively finds a non-negative flow $g \leq f$ such that $f - g \equiv_2 0$. Mumey et~al.~\cite{mumey2015parity} showed that this can be done by solving a minimum flow problem, using the following constraints: as lower bounds $0$ on edges where $f$ is even and $1$ where $f$ is odd, and as upper bounds we use $f$. In other words, they showed that the following LP formulation solves the problem for a given flow network $(G, f)$:
\begin{equation}
\begin{aligned}
&\min \val{g}, \text{subject to}\\
&g \text{ is a flow on $G$},\\
&0 \leq g \leq f,\\
&g(e) \geq 1 \text{ for all $e\in E(G)$ where $f(e)$ is odd}.\\
\label{eq:min-flow-parity-fixing}
\end{aligned}
\end{equation}
In the $i$-th iteration, starting with $i=0$, let $f_i$ be an optimal solution to the LP formulation \ref{eq:min-flow-parity-fixing} on $(G, f)$ and subtract $f_i$ from $f$. We recursively decompose $f_i$ using weight $1$ paths, which act as weight-$2^i$ paths in the decomposition of $f$. As a result, $f$ is even, and we divide it by $2$. We then follow up with iteration $i+1$ and repeat this procedure until $\val{f} = 0$. 
See \Cref{alg:mfd-approx} in \Cref{sec:pseudo-code} for a pseudo-code description.


\begin{lemma}
    Given an \mfdn input $(G, f)$, we can decompose the flow into $\lfloor\log\Vert f\Vert\rfloor+1$ flows $f_i$, such that
    \begin{equation}
    f = \sum_{i=0}^{\lfloor\log\Vert f\Vert\rfloor} 2^i f_i \text{~~~and~~~} \mfdsizeo{G}{f_i} = \val{f_i} \leq \fwidth{G}{f^{(i)}},
    \end{equation}
    where $f^{(0)} = f$ and $f^{(i)} = f - \sum_{j = 0}^{i-1} 2^j f_j$. This can be done in time $O(m^{1+o(1)}\log\Vert f\Vert)$ with high probability.
    \label{thm:approx}
\end{lemma}
\begin{proof}
    The algorithm above takes at most $\lfloor \log \Vert f\Vert \rfloor +1$ iterations until the flow is decomposed.
    We have that $\val{f_i} \leq \fwidth{G}{f^{(i)}}$, because any flow satisfying the Covering and Upper bounds constraints in \Cref{def:flow-width} is a feasible solution to LP \ref{eq:min-flow-parity-fixing}. Moreover, every $f_i$ is a minimally covering flow.

    The decomposition into minimally covering flows can be found in $O(m^{1+o(1)}\log\Vert f\Vert)$ time with high probability: LP \ref{eq:min-flow-parity-fixing} can be solved with high probability in time $O(m^{1+o(1)})$ by constructing an equivalent minimum cost flow formulation~\cite{chen2025maximum}. We solve LP \ref{eq:min-flow-parity-fixing} at most $\lfloor\log\Vert f\Vert\rfloor+1$ times.
\end{proof}

We now show that \Cref{thm:approx} improves a previous approximation ratio for \mfdn of $O(\log \val{f})$ by Cáceres et al.~\cite{caceres2024width} on width-stable graphs to $O(\log\Vert f\Vert)$.

\thmapproxalgows*
\begin{proof}
    By \Cref{thm:approx} we can express a flow $f$ as the sum of $\lfloor\log\Vert f\Vert\rfloor+1$ flows $f_i$ with $\val{f_i}\leq\fwidth{G}{f}$. Since $G$ is width-stable, and any flow decomposition is a path cover, we have $\fwidth{G}{f^{(i)}} \leq \fwidth{G}{f} \leq \mfdsizen{G}{f}$. We can thus decompose all $\lfloor\log\Vert f\Vert\rfloor+1$ flows $f_i$ with at most $\mfdsizen{G}{f}$ paths, which gives the approximation ratio.

    After expressing the sum in time $O(m^{1+o(1)}\log\Vert f\Vert)$ by \Cref{thm:approx}, we decompose each flow $f_i$ with weight $1$ paths, which takes time $O(\val{f_i} \cdot n) \leq O(\width{G}\cdot n) \leq O(nm)$. We must decompose $\lfloor \log\Vert f\Vert \rfloor +1$ flows $f_i$ and obtain the time complexity $O(nm\log\Vert f\Vert)$.
\end{proof}

\section{Parallel-width for \mfdn}
\label{sec:approx-alg}

We now present a novel view on flow decompositions, which leads to the main result of the paper. We first show that, despite being used in a different context, there exists a close connection of a certain notion of directed graph minors defined by Deligkas and Meir~\cite{deligkas2017directed} to the decomposition of flows, in the sense that contractions of flow-subgraphs are equivalent to the digraph minor operations. Next, we use this connection to parameterise \mfdn by the parallel-width, which was defined by Deligkas and Meir by the largest minimal cut-set~\cite{deligkas2017directed}. They showed that the class of graphs of $\pw{G}$ at most $c$ is hereditary. We then show that the flow-width generalises the $\width{G}$ and the $\pw{G}$. This gives us an upper bound for \mfdn and consequently an improved parameterised complexity for \mfdn on graphs of constant parallel-width, which, as we show, include width-2 DAGs.

\subsection{Flows and directed graph minors}
\label{sec:subflow-minors}

Flow decompositions $(\mathcal{P}, w)$ of $(G, f)$ of size $k$ naturally construct a sequence of subgraphs \[ G|_f \supseteq G|_{f - w_1P_1} \supseteq \dots \supseteq G|_{f - w_1P_1 - \dots - w_{k-1}P_{k-1}} \supseteq G|_0 = \emptyset, \]
for $P_i \in \mathcal{P}$ and $w = (w_1,\dots,w_k)\in\N^k$.

Moreover, flow networks admit natural edge contractions. As Kloster et al.~\cite[Lemma 4.1]{kloster2018practical} have observed, due to the flow conservation property, contracting an edge $(u,v)$ with $\deg^-(v) = 1$ or $\deg^+(u) = 1$ yields a new flow network whose flow decomposition uniquely recovers the corresponding flow decomposition of the original graph. These contractions are sometimes called ``Y-to-V''\footnote{The name Y-to-V contraction originates from the drawing of the corresponding digraphs. The descender of the letter Y corresponds to the contracted edge.} and are commonly used to simplify inputs~\cite{kloster2018practical, grigorjew2023accelerating}.

\begin{definition}[\cite{deligkas2017directed}, Definition 3]
    A digraph $G'$ is a directed minor (or d-minor) of a digraph $G$ if $G'$ can be obtained from $G$ by a sequence of the following operations:
    \begin{enumerate}
        \item \textbf{Deletion.} Deleting an edge $(a,b)$ where $\deg^+(a) > 1$ and $\deg^-(b) > 1$.
        \item \textbf{Backward contraction.} Contracting an edge $(a,b)$ where $\deg^-(b) = 1$.
        \item \textbf{Forward contraction.} Contracting an edge $(a,b)$ where $\deg^+(a) = 1$.
    \end{enumerate}
    \label{def:d-minors}
\end{definition}


We now show that d-minors as defined in \Cref{def:d-minors} correspond to compact representations of the flow networks that can appear during the process of a flow decomposition. We use the following observation to obtain this correspondence.
A case of proof of \Cref{obs:d-minor-first-delete} is visualized in \Cref{fig:first-delete-then-contract}.
\begin{figure}
    \begin{subfigure}[t]{.5\linewidth}
        \centering
        \includegraphics[width=.7\textwidth]{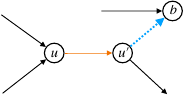}
    \end{subfigure}
    \begin{subfigure}[t]{.45\linewidth}
        \centering
        \includegraphics[width=.5\textwidth]{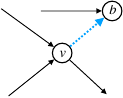}
    \end{subfigure}
    \caption{An $s$-$t$ DAG $G$, where the orange edge $(u,u')$ (left) gets contracted, obtaining the vertex $v$ (right). In both cases, the blue edge $(u', b)$ (left) or $(v,b)$ (right) can be deleted, since $\deg^+(u') > 1$, $\deg^+(v) > 1$ and $\deg^-(b) > 1$. In the context of the proof of \Cref{obs:d-minor-first-delete}, this figure displays the case $a = u'$.}
    \label{fig:first-delete-then-contract}
\end{figure}
\begin{observation}
    Let $G$ be a DAG and let $G'$ be obtained by the contraction of an edge, followed by the deletion of another edge. Then is possible to swap the order of these two operations. That is, every d-minor of a DAG can be obtained by first deleting edges, and then contracting edges.
    \label{obs:d-minor-first-delete}
\end{observation}
\begin{proof}
    Let $G$ be an $s$-$t$ DAG, let $e = (u,u')$ be an edge with $\deg^+(u) = 1$ or $\deg^-(u') = 1$ and let the DAG $G'$ be obtained by contracting $e$. In the resulting DAG $G'$ we identify $u$ and $u'$ and obtain a new vertex $v$. Let now $\tilde{e} = (a,b)$ with $\deg^+(a) > 1$ and $\deg^-(b) > 1$ be an edge to delete.
    There exists a corresponding edge $\tilde{e}'$ of $\tilde{e}$ in $G$. If $a,b \neq v$, we have $\tilde{e}' = \tilde{e}$, and we can delete $\tilde{e}'$ in $G$.
    
    
    The case $b = v$ follows equivalently, as there are two options for the corresponding edge: either $\tilde{e}' = (a, u')$ or $\tilde{e}' = (a, u)$. In both cases, $\tilde{e}'$ can be deleted from $G$: If $\tilde{e}' = (a, u')$, then $\deg^-(u') > 1$, because both $\tilde{e}'$ and $e$ are in-going edges of $u'$. If $\tilde{e}' = (a, u)$, then $\deg^-(u)$ in $G$ is equal to $\deg^-(v)$ in $G'$.
    
    Finally, after deleting $\tilde{e}'$, it remains possible to contract $e$, as deleting edges only reduces the degrees of vertices.
\end{proof}

\begin{lemma}
    A DAG $G'$ is a d-minor of a DAG $G$ if and only if $G'$ is a Y-to-V contracted graph of a flow-subgraph of $G$. \label{lem:dminor-iff-fminor}
\end{lemma}
\begin{proof}
    First assume that $G'$ is a Y-to-V contracted DAG of a flow-subgraph $H$ of $G$, such that $H = G|_f$ for some flow $f \geq 0$. To construct $H$ from $G$ using d-minor operations, we can alternate between the deletion operation and the contraction operations to delete/contract all edges with $f(e) = 0$, as we can always use at least one operation. To construct $G'$ from $H$, it is left to do Y-to-V contractions.

    Next, assume that $G'$ is a d-minor of $G$. Using \Cref{obs:d-minor-first-delete}, we can obtain $G'$ from $G$ by first deleting edges and then contracting them. Thus, edges in $G'$ are Y-to-V contractions of edges in $G$. Undoing these contractions, we obtain an $s$-$t$ DAG $H \subseteq G$. A path cover of the DAG $H$ induces a flow $f:E(G) \to \N$ that is positive on $E(H)$ and zero on $E(G - H)$, where $G - H$ is defined by $V(G-H) \coloneqq V(G)\setminus V(H)$ and $E(G-H) \coloneqq E(G)\setminus \{(u,v)\mid u\in V(H) \text{ or } v\in V(H)\}$.
\end{proof}

As a first implication, we show that the class of width-stable DAGs can be described using a forbidden minor.

\begin{definition}
    \label{def:chk}
    We define the graph $\ch{k}$ to consist of $4$ vertices $s, u, v, t$ and of the following edges: $k$ parallel edges $(s,u)$, $k$ parallel edges $(v,t)$, and the three edges $(s, v), (u, v), (u, t)$.
\end{definition}

\begin{lemma}
    Let $G$ be an $s$-$t$ DAG. $G$ is width-stable if and only if $G$ is $\ch{2}$-d-minor free.
    \label{lem:ws-forbidden-minor}
\end{lemma}
\begin{figure}
    \begin{subfigure}[t]{.45\linewidth}
        \centering
        \includegraphics[width=.5\textwidth]{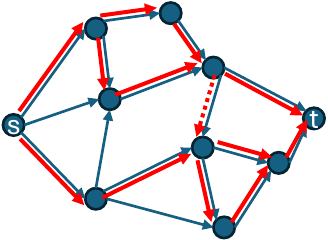}
    \end{subfigure}
    \begin{subfigure}[t]{.45\linewidth}
        \centering
        \includegraphics[width=.5\textwidth]{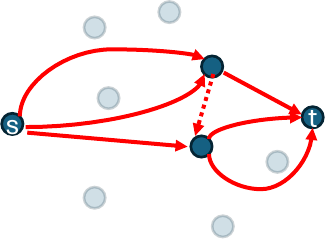}
    \end{subfigure}
    \caption{An $s$-$t$ DAG that is not width-stable. The flow-subgraph in red can be Y-to-V contracted to $\ch{2}$.}
    \label{fig:widthstable}
\end{figure}
\begin{proof}
    If $G$ contains $\ch{2}$ as d-minor, then it contains a flow-subgraph $H$ such that $\ch{2}$ is a Y-to-V contracted graph of $H$. $H$ is then exactly a funnel (of maximum minimal cut-set size $4$) with a central path, and by \Cref{lem:wstable-funnel-char} $G$ is not width-stable.

    If $G$ is not width-stable, it contains by \Cref{lem:wstable-funnel-char} a funnel $F$ with central path $P$ from $u$ to $v$, where $\deg^-(u)\geq 2$ and $\deg^+(v)\geq 2$ with respect to $F$. This means that there are at least two distinct paths from $s$ to $u$ and two distinct paths from $v$ to $t$. Since $F$ is a funnel, there also exists a path from $u$ to $t$ avoiding $v$, and a path from $s$ to $v$ avoiding $u$.
    This subgraph of paths has a minimum path cover of size $3$, and it induces a flow $f$ on $G$, such that $\ch{2}$ is a Y-to-V contracted graph of $G|_f$.
\end{proof}

An example of \Cref{lem:ws-forbidden-minor} is illustrated in \Cref{fig:widthstable}.
The lemma gives a natural proof showing that width-stable DAGs generalise series-parallel DAGs, as series-parallel DAGs are precisely $\ch{1}$-d-minor free DAGs~\cite{holzman2003network}.

Detecting minors of constant size can be done in polynomial time, which was shown using several previous results in graph minor theory~\cite{deligkas2017directed}. However, there is a simple polynomial algorithm that detects if $\ch{2}$ is present in a DAG $G$ as d-minor, which works by computing reachability questions on $G$.
\begin{corollary}
    \label{cor:detect-ch2}
    There exists an algorithm that detects whether an $s$-$t$ DAG $G$ is width-stable in time $O(n^3)$.
\end{corollary}
\begin{proof}
    First, we compute the number of $s$-$v$ paths $d_s(v)$ and the number of $v$-$t$ paths $d_t(v)$ for every $v\in V(G)$, which we can do in time $O(n^2)$ (since the size of these integers is at most $2^n$). Next, we iterate over all pairs of two internal vertices  $u < v$ for a topological order $<$ of $G$, for which $d_s(u)\geq 2$ and $d_t(v)\geq 2$. This ensures that there is a d-minor with two parallel edges from $s$ to $u$ and two parallel edges from $v$ to $t$.
    With a graph search from $u$, we can check whether there exists a path to $v$. Finally, we must check if there exists a path from $s$ to $v$ that avoids $u$. We can do so by removing $u$ and all its edges from $G$ and by doing a graph search from $s$ in that subgraph. Similarly, we can check if there exists a path from $u$ to $t$ that avoids $v$.
    If all these paths exist, we have shown that there is a $\ch{2}$ d-minor with internal vertices $u$ and $v$. We iterate through $O(n^2)$ pairs and perform a constant number of linear graph searches, obtaining a total runtime of $O(n^3)$.
\end{proof}

\subsection{\mfdn parameterised upper bounds}
\label{sec:mfd-approx-parallel-width}

Since the width of a flow-subgraph of a DAG $G$ can be larger than the width of $G$, we use a structural parameter to describe the class of DAGs whose flow-subgraphs' widths stay below a given threshold. It is the third parameter in this paper, and is connected to the previous two, the flow-width of $(G,f)$ and the width of $G$.
\begin{definition}[\cite{deligkas2017directed}, Definition 12]
    The \textit{parallel-width} $\pw{G}$ of an $s$-$t$ DAG $G$ is the size of the  largest minimal cut-set of $G$. That is, the size of a set $C \subseteq E$, such that every $s$-$t$ path crosses $C$, and in addition, for every $C' \subsetneq C$, there exists some $s$-$t$ path that avoids $C'$.
    \label{def:parallel-width}
\end{definition}
Let $\gpc{c}$ be the DAG consisting of two nodes $s$ and $t$ and $c$ parallel edges $(s,t)$. Deligkas and Meir~\cite{deligkas2017directed} showed that it is $\mathsf{NP}$-hard to compute the parallel-width and that the class of DAGs $G$ with $\pw{G} < c$ for a constant $c\in\N$ are characterized by the forbidden d-minor $\gpc{c}$. The following lemma shows why it is relevant for the parity fixing algorithm.

\lempwidthwidthfwidth*
\begin{proof}
We need to show four inequalities:
\begin{enumerate}
    \item $\width{G} \leq \min\{\fwidth{G}{f} \mid f > 0\}$: $\fwidth{G}{f}$ is the value of some positive flow, whereas the $\width{G}$ is the minimum value of any positive flow.
    \item $\width{G} \geq \min\{\fwidth{G}{f} \mid f > 0\}$: Consider $f$ to be the induced flow of a minimum path edge cover of $G$. Then $\fwidth{G}{f} = \val{f} = \width{G}$.
    \item $\pw{G} \leq \max\{\fwidth{G}{f} \mid f \geq 0 \}$: Let $C = \{ (u_1, v_1), \dots, (u_\ell, v_\ell) \}$ be the largest minimal cut-set. It was shown in~\cite{deligkas2017directed} that there exists an out-tree from $s$, with leaves $u_i$ for $i\in[\ell]$ and an in-tree from the leaves $v_i$ for $i\in[\ell]$ to the root $t$. This is exactly a funnel subgraph $F$, and the width of it is $\pw{G}$, since $C$ is the set of the private edges of $F$. To verify the inequality, we can choose for $f$ the induced flow of the unique path cover of $F$.
    \item $\pw{G} \geq \max\{\fwidth{G}{f} \mid f \geq 0\}$: The right hand side is the maximum value of all non-negative minimally covering flows $f$. By definition of minimally covering flows, there exists a minimal cut-set $C$ in $G|_f$ such that every edge in $C$ is covered at most once. We have $\pw{G} \geq |C| \geq \val{f} = \fwidth{G}{f}$.
\end{enumerate}
\end{proof}

\begin{corollary}
    If an $s$-$t$ DAG $G$ is width-stable, then $\width{G} = \pw{G}$.
    \label{cor:width-eq-pw-if-ws}
\end{corollary}
\begin{proof}
   Clearly, for all DAGs, we have $\width{G} \leq \pw{G}$. Let $f$ and $h$ be flows such that $\fwidth{G}{f} = \pw{G}$ and $\fwidth{G}{h} = \width{G}$ according to \Cref{lem:pwidth-width-fwidth}. By \Cref{lem:fw-first-properties}, since $h > 0$, we have $\fwidth{G}{f+h} \leq \fwidth{G}{h}$, and thus $\fwidth{G}{f+h} = \width{G}$. By \Cref{lem:fwstable}, if $G$ is width-stable, we have $\pw{G} = \fwidth{G}{f} \leq \fwidth{G}{f+h} = \width{G}$.
\end{proof}

The following Lemma is essential for showing the main result of the paper.

\thmapproxalgogeneral*
\begin{proof}
    By \Cref{lem:pwidth-width-fwidth}, we have $\fwidth{G}{f}\leq\pw{G}$ for any flow $f\geq 0$. Hence, the parity-fixing algorithm returns at most $\pw{G}\cdot(\lfloor\log\Vert f\Vert\rfloor + 1)$ many paths. 

    For the runtime, as before, we express $f$ in time $O(m^{1+o(1)}\log\Vert f\Vert)$ as the sum of $\lfloor\log\Vert f\Vert\rfloor + 1$ flows $f_i$. We have $\val{f_i} \leq \pw{G} \leq m$ for all $f_i$, and thus take $O(nm)$ time to decompose one $f_i$.
\end{proof}

Since for $f> 0$ we have $\width{G} = \width{G|_f} \leq \mfdsizen{G}{f}$, the approximation ratio is $\frac{\pw{G}}{\mfdsizen{G}{f}}\cdot(\lfloor\log\Vert f\Vert\rfloor + 1) \leq \frac{\pw{G}}{\width{G}}\cdot(\lfloor\log\Vert f\Vert\rfloor + 1)$. The fraction $\frac{\pw{G}}{\width{G}}$ generalises width-stability, and describes the factor by which the width can possibly grow during the decomposition of a flow. 
If $\mfdsizen{G}{f} \geq \pw{G}$, we moreover obtain a ratio of $\lfloor\log\Vert f\Vert\rfloor + 1$.
In general, the values $\fwidth{G}{f^{(i)}}$ can grow above $\mfdsizen{G}{f}$.
In \Cref{sec:large-apx-ratio}, we show that the approximation ratio of the parity-fixing heuristic can be as large as $\Omega(\Vert f\Vert)$ in some classes of graphs, which shows a contrast between \mfdn and \mfdz.
The improved upper bound of the algorithm analysis implies the following theorem.

\corconstantparwunary*
\begin{proof}
    \Cref{thm:approx-algo-general} yields an upper bound for $\mfdsizen{G}{f}$ of size $\pw{G}\cdot (\lfloor\log\Vert f\Vert\rfloor+ 1)$. It has previously been shown that \mfdn is in FPT~\cite{kloster2018practical} with parameter $k = \mfdsizen{G}{f}$, with an implemented tool \textsc{toboggan} that runs in time $4^{k^2} k^{1.5k} k^{o(k)} 1.765^k \cdot (n + \log\Vert f\Vert) = 2^{O(k^2)} \cdot (n + \log\Vert f\Vert)$~(\cite{kloster2018practical}, Theorem 7). Substituting $k$ with the upper bound, we obtain a runtime dependent on the $\pw{G}$ and the logarithm of the largest flow weight in the exponent. If we assume that $\pw{G}$ is at most a constant, this running time is $\Vert f\Vert^{O(\log\Vert f\Vert)} \cdot (n + \log\Vert f\Vert) = \Vert f\Vert^{O(\log\Vert f\Vert)} \cdot n$.
\end{proof}

\Cref{cor:constant-parw-unary} implies that \Cref{thm:constant-parw-binary} is likely tight, as a reduction showing \textit{strong} \NP-hardness on width-2 graphs does not exist, unless $\mathsf{NP} \subseteq \mathsf{QP}$, as we now show:
\begin{lemma}
    Let $G$ be an $s$-$t$ DAG with $\width{G} \leq 2$. Then $G$ is width-stable.
    \label{lem:width-2-width-stable}
\end{lemma}
\begin{proof}
    We show with induction on the edge set, that every $s$-$t$ DAG $G$ with $\width{G} \leq 2$ is width-stable using \Cref{def:d-minors}. $G$ is obviously width-stable if $E(G) = \{(s,t)\}$.
    
    Let $G$ be any $s$-$t$ DAG with $\width{G} \leq 2$ and let $(u,v) \in E(G)$.
    If $\deg^+(u) > 1$ and $\deg^-(v) > 1$, we can remove the edge $(u,v)$. But since $\width{G} \leq 2$, there must be a $u$-$v$ path disjoint from the edge $(u,v)$. When removing $(u,v)$, the width of the new DAG is not strictly larger, because we can route both paths of a minimum path cover through the $u$-$v$ path.
    Moreover, edge contractions do not change the width of the DAG. This is the case, because minimum path covers are flow decompositions of a minimum flow, and Kloster et al.~\cite[Lemma 4.1]{kloster2018practical} showed that $\mfdsizen{G}{f} = \mfdsizen{G'}{f'}$, where $G'$ and $f'$ are obtained from $G$ and $f$ by contracting an edge.
\end{proof}

\begin{corollary}
    \mfdn on $(G,f)$ can be solved in time $\Vert f\Vert^{O(\log\Vert f\Vert)}\cdot n$ when $\width{G} = 2$, which is quasi-polynomial if $f$ is coded in unary.
    \label{cor:width-2-unary-qp}
\end{corollary}
\begin{proof}
    The statement follows from \Cref{cor:width-eq-pw-if-ws}, \Cref{cor:constant-parw-unary} and \Cref{lem:width-2-width-stable}. 
\end{proof}

Finally, we partially address an open question stated by Kloster et al.~\cite{kloster2018practical}, which directly follows from the reduction in \Cref{thm:constant-parw-binary} and \Cref{cor:width-2-unary-qp}:
\begin{corollary}
    If \mfdn can be solved in time $O^*(2^{o(k^2)})$, where $k = \mfdsizen{G}{f}$, we can solve \textsc{generating set} in time $O^*(2^{o(\log(s)^2)})$, where $s$ is the maximum integer in the input.
\end{corollary}

This result shows that improving the time $O^*(2^{O(k^2)})$ for \mfdn is inherently an additive combinatorics problem, rather than a graph theoretic problem.

\section{Conclusions}
In this article we answered problems regarding the computational complexity of \mfdn parameterised by DAG width notions. We showed that \mfdn on graphs $G$ with $\width{G} = 2$ or $\pw{G} \leq c$ obtains a quasi-polynomial runtime when the flow is coded in unary and is NP-hard for binary coded flows. This shows that \mfdn with constant parallel-width naturally generalises the \textsc{generating set} problem. When the width is equal to 3, the problem remains even strongly NP-hard.

We achieved this result using two techniques. First, by introducing the flow-width of a flow network, we relate minimally covering flows to graph structure. Second, by connecting flow decompositions to directed minors. We showed an intermediate result that expressing the input flow as the sum of minimally covering flows can lead to efficient upper bounds for minimum flow decompositions.

An interesting open question is, whether \textsc{generating set} can be solved in time $O^*(2^{o(\log(s)^2)})$, where $s$ is the maximum integer in the input, or can the existence of such an algorithm be ruled out under classical computational complexity assumptions? This paper shows that a negative answer would also show that there is no $O^*(2^{o(k^2)})$ time algorithm for \mfdn.

\bibliography{main}

\appendix

\section{Pseudo-code}
\label{sec:pseudo-code}
\begin{algorithm}
    \caption{Parity-fixing heuristic for \mfdn~\cite{mumey2015parity}}
    \label{alg:mfd-approx}
    \begin{algorithmic}[1]
        \REQUIRE{\mfd instance $(G,f)$}
        \ENSURE{Flow decomposition $\mathcal{P}$}
        \STATE $i \gets 0$
        \STATE $\mathcal{P} \gets \{\}$
        \WHILE{$f > 0$}
            \STATE $h \gets \text{Optimal solution of LP \ref{eq:min-flow-parity-fixing}}$ \label{line:min-covering-flow-of-odd-edges}
            \STATE $\{P_{i,1},\dots,P_{i,\val{h}}\} \gets \text{FD}(G, h)$ \text{Flow decomposition into weight $1$ paths} \label{line:fd-parity-fixing}
            \STATE $\{w_{i,1},\dots,w_{i,\val{h}}\} \gets \{2^i,\dots,2^i\}$
            \STATE $\mathcal{P} \gets \mathcal{P} \cup \{ (P_{i,j}, w_{i,j})\}$ for $j = 1,\dots,\val{h}$
            \STATE $f \gets (f - h) / 2$ \label{line:update-f}
            \STATE $i \gets i+1$
        \ENDWHILE
    \end{algorithmic}
\end{algorithm}


%
%
%

\section{Larger approximation ratio}
\label{sec:large-apx-ratio}

\begin{figure}[ht]
    \centering
    \includegraphics[width=.75\textwidth]{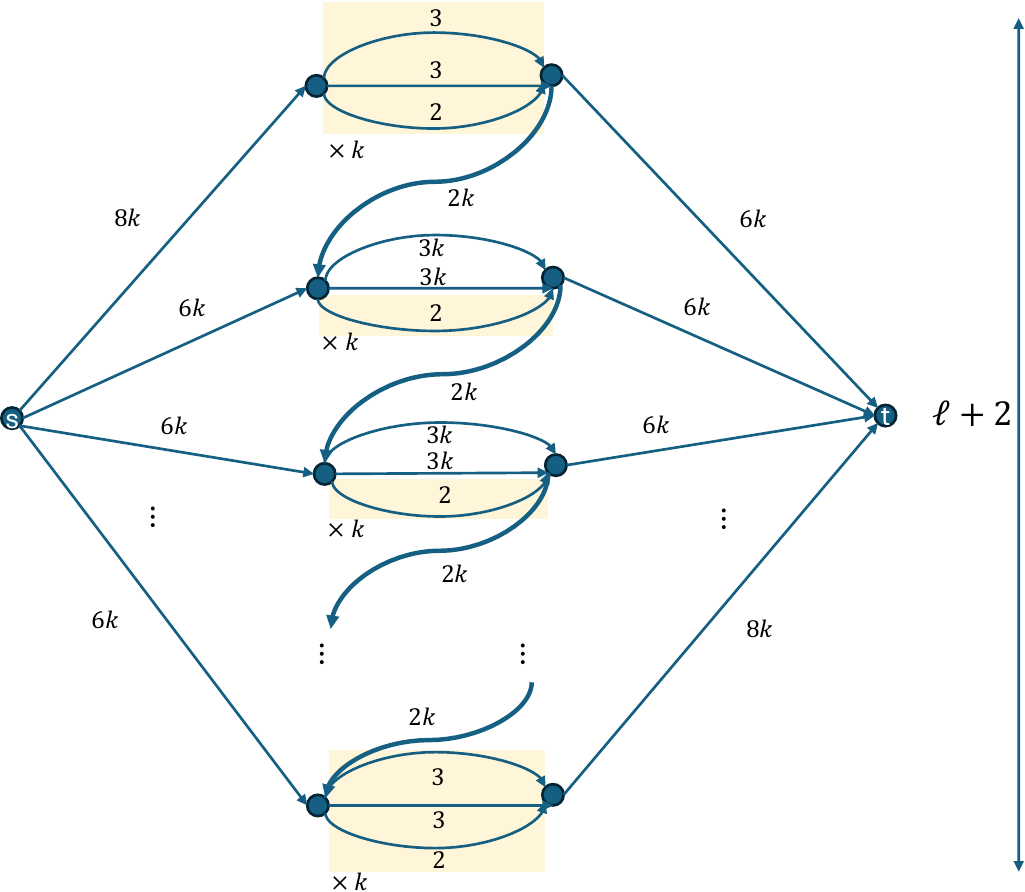}
    \caption{\mfd instance $(G_{(k, \ell)}, f_k)$. The yellow highlighted edges are copied $k$ times, and there are $\ell+2$ total central gadgets (that is, $\deg^+(s) = \ell+2$). Thus, $\pw{G_{(k, \ell)}} = 6k + \ell(k+2)$, and we call these edges the maximum cut-set edges. There are two central gadgets consisting of $4k$ maximum cut-set edges and $\ell$ central gadgets each consisting of $k+2$ maximum cut-set edges.}
    \label{fig:bad-example-approx}
\end{figure}
Consider the \mfdn instance $(G_{(k, \ell)}, f_k)$ defined in \Cref{fig:bad-example-approx}. The following lemma shows the idea of the construction.

\begin{lemma}
    \label{lem:large-pw-mfd-ratio}
    For all $c > 1$, there exist $k, \ell > 0$, such that $$\frac{\pw{G_{(k, \ell)}}} {\mfdsizen{G_{(k, \ell)}}{f_k}} > c.$$
\end{lemma}
\begin{proof}
    Clearly, $\pw{G_{(k, \ell)}} = 6k + \ell(k+2)$.
    To decompose the flow in an efficient way, we can decompose all the maximum cut-set edges of flow $2$ with $k$ paths, which also saturates all the connecting central edges between them (in bold). We then use additional $4k + 2\ell$ paths to fully decompose the graph. In total, $\mfdsizen{G_{(k, \ell)}}{f_k}\leq 5k + 2\ell$, and thus, $$ \frac{\pw{G_{(k, \ell)}}}{\mfdsizen{G_{(k, \ell)}}{f_k}} \geq \frac{6k + \ell(k+2)}{5k + 2\ell} \xrightarrow{k \to \infty} \frac{6 + \ell}{5}. $$
    This shows that for $\ell > 5c - 6$, there exists $k > 0$, such that $\pw{G_{(k, \ell)}} / \mfdsizen{G_{(k, \ell)}}{f_k} > c$.
\end{proof}

\sloppy The strategy is now to show that the approximation algorithm uses in the worst case almost $\pw{G_{(k, \ell)}}$ many paths to cover the odd edges in the second iteration. This will imply the approximation factor of $\Omega(\Vert f\Vert)$. Let $\mathcal{P}_i$ be the set of paths that the approximation algorithm uses to cover the odd edges in iteration $i$. That is, $|\mathcal{P}_i|$ is the optimal solution of the LP \ref{eq:min-flow-parity-fixing} in iteration $i$ of the algorithm.

\begin{lemma}
    \label{lem:mfd-approx-bad-in-general}
    In the worst case, \Cref{alg:mfd-approx} is a factor $\Omega(||f||)$ approximation algorithm.
\end{lemma}
\begin{proof}
    Let $k$ be odd. In the first iteration the approximation algorithm fixes the parity of the $4k$ edges whose flow weight is each $3$ using $2k$ many paths. If they traverse the $2\ell$ edges whose flow weight is each $3k$, they decompose the central connecting edges in bold of flow $2k$, and only them. In the second iteration ($i = 1$), after dividing the flow by $2$, there are $|\mathcal{P}_1| = 6k + k\ell$ many odd edges that are pairwise unreachable. Thus, if we have $k = \ell$,
    $$ \frac{|\mathcal{P}_1|}{\mfdsizen{G_{(k, k)}}{f_k}}\geq \frac{6k + k^2}{6k} = \Theta(k). $$
    The approximation factor follows, because $k = \frac{1}{8}\Vert f\Vert = \Theta(\Vert f\Vert)$, and the number of paths returned by the approximation algorithm is at least $|\mathcal{P}_1|$.
\end{proof}

\end{document}